\documentclass[12pt, a4paper, reqno]{amsart}
\usepackage{mathrsfs}
\usepackage{bbm}
\usepackage{pgf,tikz}
\usepackage{tabu}
\usepackage{amsmath,amscd,amssymb,latexsym}
\usepackage[all]{xy}

\usepackage{resizegather}
\usepackage{color}

\input xypic

\evensidemargin=\oddsidemargin
\DeclareMathVersion{can}
\DeclareMathAlphabet{\can}{OT1}{cmss}{m}{n}
\newtheorem{thm}{Theorem}[section]
\newtheorem{cor}[thm]{Corollary}
\newtheorem{lem}[thm]{Lemma}
\newtheorem{prop}[thm]{Proposition}

\newtheorem{exa}[thm]{Example}
\theoremstyle{definition}

\theoremstyle{fact}

\theoremstyle{conjecture}

\numberwithin{equation}{section}

\begin{document}
\title[quaternary codes  and their   binary images]
{quaternary codes  and their   binary images}

\author[Y. Wu]{Yansheng ~Wu}

\address[Y. Wu]{ School of Computer Science, Nanjing University of Posts and Telecommunications, Nanjing 210023, P. R. China.}
\email{yanshengwu@njupt.edu.cn}

\author[C. Li]{Chao Li}
\address[C. Li]{ School of Computer Science, Nanjing University of Posts and Telecommunications, Nanjing, 210023, P. R. China}
\email{1022041222@njupt.edu.cn}

\author[L. Zhang]{Lin Zhang}
\address[L. Zhang]{ School of Computer Science, Nanjing University of Posts and Telecommunications, Nanjing, 210023, P. R. China}
\email{zhangl@njupt.edu.cn}

\author[F. Xiao]{Fu Xiao}
\address[F. Xiao]{School of Computer Science, Nanjing University of Posts and Telecommunications, Nanjing 210023, P. R. China}
\email{xiaof@njupt.edu.cn}


\date{\today}


\baselineskip=20pt

\begin{abstract}   

Recently, simplicial complexes are used in constructions of several infinite families of  minimal and optimal linear codes by Hyun {\em et al.}  Building upon their research, in this paper more linear codes over the ring $\mathbb{Z}_4$ are constructed by simplicial complexes. Specifically, the Lee weight distributions of the resulting quaternary codes are determined and  two infinite families of four-Lee-weight quaternary codes are obtained. Compared to the databases of $\mathbb Z_4$ codes by Aydin {\em et al.}, at least nine new quaternary codes are found.  Thanks to the special structure of the defining sets, we have the ability to determine whether the Gray images of certain obtained quaternary codes are linear or not. This allows us to obtain two infinite families of binary nonlinear codes and  one infinite family of binary minimal linear codes.  Furthermore, utilizing these minimal binary codes, some secret sharing schemes as a byproduct also are established.


\end{abstract}

\maketitle

\bigskip
\section{Introduction}

Research on codes over rings {began as early as in the 1970s}. Hammons {\em et al.} \cite{HTZ} constructed some quaternary linear codes, including the quaternary Kerdock, Preparata, Goethals, Delsarte-Goethals, and Goethals-Delsarte codes, and gave some well-known nonlinear binary codes with recognizing error-correcting capabilities
which can be identified as images of linear codes over $\mathbb{Z}_4$ under the Gray map. In the book titled  Quaternary Codes \cite{WQ}   by  Z. X. Wan published in 1997, various aspects of quaternary codes {were} discussed. Since then, the constructions and properties of quaternary codes have attracted attention from scholars.

Recently, Hyun {\em et al}. \cite{CH,HLL} constructed some infinite families of binary optimal and minimal linear codes using simplicial complexes. And it is generalized to posets in \cite{HKWY}. After then, Wu {\em et al.} \cite{WZY} used simplicial complexes to present some optimal few-weight codes. Zhu {\em et al.} \cite{ZWY} utilized posets of the disjoint union of two chains to construct quaternary codes and presented  several new quaternary codes.  Wu {\em et al.} \cite{WLX} examined the connection between quaternary codes and their binary subfield codes. Specifically, they investigated how defining sets of quaternary codes relate to their binary subfield codes, and discovered some optimal quaternary codes as well as optimal binary subfield linear codes.


In this paper, we focus on the construction of  linear codes $\mathcal{C}$ over  $\mathbb{Z}_4$ as follows:
\begin{equation}\label{eq1}
  \mathcal C_{D}=\{c_{D}(\mathbf{u})=(\mathbf{u}\cdot \mathbf{t})_{\mathbf{t}\in D}:\mathbf u\in \Bbb Z_4^n\},
\end{equation}
where $D\subseteq \Bbb Z_4^n$ is called the defining set of the quaternary code $\mathcal C_D$. Here we will choose the defining set $D$ using simplicial complexes.

The rest of this paper is orgnaized as follows. In Section 2, we introduce some basic concepts and known results on simplicial complexes, Lee weight, the Gray map, and minimal codes. In Section 3, we determine the Lee weight distribution of the quaternary codes $\mathcal {C}_D$ in (\ref{eq1}). Some of them have four or five Lee weights.  In Section 4, we determine when the quaternary codes obtained in Section 3 have linear Gray images, and we also find two infinite families of binary nonlinear codes, give a class of  minimal binary codes  and give examples to secret sharing schemes. Finally, we compare our quaternary codes with known results and conclude the paper.

\section{Preliminaries}  

In this section, we will recall some concepts and results on linear codes, Gray map, quaternary codes,  simplicial complexes,  and minimal codes.

\subsection{Linear codes over finite fields}  Let $n$ be a positive integer and $\Bbb F_q^n$ denote the vector space of all $n$-tuples over the ﬁnite ﬁeld $\Bbb F_q$. We call $\mathcal{C}$ {an} $[n,k,d]$ linear code over $\Bbb F_q$ if $\mathcal{C}$ is a subspace of $\Bbb F_q^n$ of dimension $k$ with minimum Hamming distance $d$.  We also call the vectors in $\mathcal C$ codewords.
For  an integer $i $ with  $1\le i \le n$, assume that there are $A_i$ codewords in $\mathcal C$ with Hamming weight $i$.
Then $(1, A_1, \ldots, A_n)$  and $1+A_1z+\cdots +A_nz^n\in \mathbb Z[z]$ are called weight distribution and  weight enumerator of $\mathcal C$, respectively.
 Moreover, the code $\mathcal C$  {is said to be} {\it $t$-weight} if the number of nonzero $A_{i}$'s in the sequence $(A_1, \ldots, A_n)$ is exactly equal to $t$. We say that  the $[n,k,d]$ code  $\mathcal{C}$ is  {\it distance optimal} if there is no $[n,k,d+1]$ code (refer to~\cite[Chapter 2]{HP}). 
Given a linear code $\mathcal C$ of length $n$ over $\Bbb F_q$,
the  dual code of $\mathcal C$  is defined by
$$\mathcal C^{\bot} = \{ {\bf x}\in \Bbb F_q^{n} \mid {\bf x}{\bf y} ^T=0\mbox{ for  all } {\bf y}\in \mathcal C\}$$
where ${\bf x}{\bf y} ^T$ {denotes} the standard inner product of two vectors ${\bf x}$ and ${\bf y}$. 
 
 

 

\subsection{Lee weight, Gray map and quaternary codes}

Let $\mathbb{Z}_4$ be a ring of integers moudule 4, $n$ be a positive integer, and $\mathbb{Z}_4^n$ be the set of $n$-tuples over $\mathbb{Z}_4$. Every non-empty subset $\mathcal{C}$ of $\mathbb{Z}_4^n$ called a quaternary code. A linear code $\mathcal{C}$ of length $n$ is a submodule of $\mathbb{Z}_4^n$.

For each $u \in \mathbb{Z}_4$ there is a unique represenation $u = a + 2b$, where $a,b \in \mathbb{Z}_2$, the Gray map $ \phi $ from $\mathbb{Z}_4$ to $\mathbb{Z}_2^2$ is {defined}  by:
$$\phi : \mathbb{Z}_4 \to \mathbb{Z}_2^2,a+2b \mapsto (b,a+b).$$
Any vector $u \in \mathbb{Z}_4^n$ can be written as $\mathbf{u} = \mathbf{a} + 2\mathbf{b}$, where $\mathbf{a},\mathbf{b} \in \mathbb{Z}_2^n$. The map $\phi$ can be extended naturally from $\mathbb{Z}_4^n$ to $\Bbb Z_2^{2n}$ as follows:
 $$\phi: \mathbb{Z}_4^n\to \mathbb{Z}_2^{2n},\mathbf{u}=\mathbf{a}+2\mathbf{b}\mapsto (\mathbf{b}, \mathbf{a}+\mathbf{b}).$$
The Hamming weight of a vector {$\bf a$}  of length $n$ over $\mathbb{Z}_2$ is defined to be the number of nonzero entries in the vector {$\bf a$}. The Lee weight of a vector {$\bf u$} of length $n$ over $\mathbb{Z}_4$ is defined to be the Hamming weight of its Gray image as follows:
$$w_{L}(\mathbf{u}) = w_{L}(\mathbf{a}+2\mathbf{b}) = w_{H}(\mathbf{b}) + w_{H}(\mathbf{a}+\mathbf{b}).$$

The Lee distance of $\mathbf{x},\mathbf{y}\in \mathbb{Z}_4^n$ is defined as $w_L(\mathbf{x-y})$. {From}, \cite[theorem 3.1]{WQ},  the Gray map is an isometry from $(\mathbb{Z}_4^n, d_L)$ to $(\mathbb{Z}_2^{2n}, d_H)$ and  is a weight-preserving map. 
In general, the Gray image of a quaternary code is not necessarily a binary linear code, since the Gray map is not $\mathbb{F}_2$-linear. Assume that a quaternary code $\mathcal C$ has Lee weights $\{i_1, \ldots, i_k\}$ and 
there are $A_i$ codewords in a quaternary $\mathcal C$ with the Lee weight $i$.
Then $(1, A_{i_1}, \ldots, A_{i_k})$  and $1+A_{i_1}z^{i_1}+\cdots +A_{i_k}z^{i_k}\in \mathbb Z[z]$ are called Lee weight distribution and Lee weight enumerator of $\mathcal C$, respectively.

Two codes are said to be equivalent if one can be obtained from the other by permuting the coordinates and (if necessary) changing the signs of certain coordinates. Codes differing by only a permutation of coordinates are called permutation-equivalent. Any linear code $\mathcal C$ over $\mathbb Z_4$ is permutation-equivalent to a code with generator matrix $G$ of the form
\begin{eqnarray} \left( \begin{array}{ccccccc}
I_{k_1} &A_1&  B_1+2B_2 \\
0& 2I_{k_2}&2A_2
\end{array} \right),
\end{eqnarray}
where $A_1 , A_2 , B_1$, and $B_2$ are matrices with entries $0$ or $1$ and $I _k$ is the identity matrix of order $k$. Such a code has size $4^{k_1}2^{k_2}$ . From \cite{WQ}, the code is a free module if and only if $k_2 = 0$. If $\mathcal C$ has length $n$ and minimum Lee weight $d_L$, then it is referred to as an $[n, 4^{k_1} 2^{k_2}, d_L ]$-code.

\subsection{Simplicial complexes}

Let $\mathbb{Z}_2$ be the ring with order two and  $n$ be a positive integer. The support $\mathrm{supp}({\bf v})$ of a vector ${\bf v} \in \mathbb{Z}_2^n$ is defined by 
the set of nonzero coordinate positions. The Hamming weight $w_H({\bf v})$ of ${\bf v}\in \mathbb{Z}^n_2$ is defined by the size of $\mathrm{supp}({\bf v})$. 
There is a bijection between $\mathbb{Z}_2^n$ and $2^{[n]}$ being the power set of $[n]=\{1, \ldots, n\}$, defined by ${\bf v}\mapsto$ supp$({\bf v})$.

For two vectors ${\bf u,v}\in \mathbb{Z}_2^n$, we say ${\bf v}\subseteq {\bf u}$ if $\mathrm{supp}({\bf v})\subseteq \mathrm{supp}({\bf u})$.  We say that a 
family $\Delta \subseteq \mathbb{Z}_2^n $ is a {\em simplicial complex} if ${\bf u}\in \Delta$ and ${\bf v}\subseteq {\bf u}$ imply ${\bf v}\in \Delta$. For a simplicial complex $\Delta$, a maximal element of $\Delta $ is one that is not properly contained in any other element of $\Delta$. Let $\mathcal{F}=\{F_1, \ldots, F_n\}$ be the family of maximal elements of $\Delta$. 
For each $F\subseteq [n]$, the simplicial complex $\Delta_F$ generated by $F$ is defined to be the family of all subsets of $F$. Throughout this paper, we will identify a vector in $\mathbb{Z}_2^n$  with its support. 

Let $\Delta \subseteq \mathbb{Z}_2^n $ be a simplicial complex. Hyun {\em et al}. \cite{CH} introduced the following $n$-variable generating function associated with the set $\Delta $:
$$ \mathcal{H}_{\Delta}(x_1,x_2,\ldots, x_n)=\sum_{{\bf v}\in \Delta}\prod_{i=1}^nx_i^{v_i}\in \mathbb{Z}[x_1,x_2, \ldots, x_n], $$
where ${\bf v}=(v_1,v_2,\ldots ,v_n)  \in \mathbb{Z}_2^n$.





The following lemma is crucial for determining the Lee weight distributions of the codes in \eqref{eq1}.

\begin{lem}\cite[Theorem 1]{CH}  \label{lem2.1} {\rm  
Let $\Delta$ be a simplicial complex of $\mathbb{Z}_2^n$ with the set of maximal elements $\mathcal {F}$. Then we have
$$ \mathcal{H}_{\Delta}(x_1,x_2,\ldots, x_n)=\sum_{\emptyset\neq S\subseteq \mathcal{F}}(-1)^{|S|+1}\prod_{i\in \cap S}(1+x_i), $$
where we define $\prod_{i\in \emptyset}(1+x_i) = 1$, and the size of $\Delta$ 
$$ |\Delta| = \sum_{\emptyset \neq S \subseteq \mathcal{F}}(-1)^{|S|+1}2^{|\cap S|}.$$}
\end{lem}

\subsection{Minimal linear codes}

For two vectors ${\bf u,v}\in \Bbb F_2^n$, we say that ${\bf u}$ covers ${\bf v}$ if $\mathrm{supp}({\bf v})\subseteq\mathrm{supp}({\bf u})$. 
A nonzero codeword ${\bf u}$ in a linear code $\mathcal{C}$ is said to be {\em minimal} if ${\bf u}$ covers the zero vector and the ${\bf u}$ itself but no other codewords in the code $\mathcal{C}$. A linear code $\mathcal{C}$ is said to be {\em minimal} if every nonzero codeword in  $\mathcal{C}$ is minimal.

A known sufficient condition for a linear code to be minimal {is called the the Ashikhmin-Barg condition}. More precisely we have the following.

\begin{lem}\cite{AB}\label{lem2.2} {\rm 
Let $\mathcal{C}$ be a linear code over $\mathbb{F}_q$ with $wt_{min}$ and $wt_{max}$ as minimum and maximum weights of its non-zero codewords. If 
$$\frac{wt_{min}}{wt_{max}} > \frac{q-1}{q},$$ 
then $\mathcal{C}$ is minimal.
}
\end{lem}

The first example of a minimal linear code violating Ashikhmin-Barg's condition was given by  Cohen {\em et al.} in \cite[Remark 1]{CMP}. Chang {\em et al.} \cite{CH} presented one infinite family of minimal binary linear codes violating Ashikhmin-Barg's condition. Ding, Heng and Zhou \cite{DHZ, HDZ} presented a necessary and sufficient condition for $q$-ary linear codes to be minimal, and obtained some infinite families of minimal binary and ternary linear codes not satisfying the condition of Ashikhmin-Barg.


An important application of minimal linear codes is to build secret sharing schemes
\cite{YD,M,LY}, which are an important cryptographic primitive that enables to distribute a secret among multiple participants. According to \cite{YD,DD}, interesting access structures can be derived from the dual code $\mathcal{C}^{\bot}$ of a linear code $\mathcal{C}$. It is widely used in banking systems, cryptographic protocols and the control of nuclear weapons. Secret sharing schemes constructed from minimal linear codes are known to possess access structures. According to \cite{YD}, the secret sharing  scheme base on a $q$-ary $(N,K)$ linear code $\mathcal{C}$ over $\mathbb{F}_q$, the secret $s$ is an element of $\mathbb{F}_q$. The  scheme involves $N-1$ participants and one trusted dealer. The minimal codewords in the dual code $\mathcal{C}^{\bot}$, where the first component is specified as $1$, define the access structure for the secret-sharing scheme corresponding to the linear code $\mathcal{C}$.

\section{Lee Weight Distributions of Quaternary Codes}


In this section,  we focus on quaternary codes over $\mathbb{Z}_4$ by  employing simplicial complexes with two maximal elements.


Let $D_i \subseteq \Bbb Z_2^n$, $i = 1,2$. Assume that $D = D_1 + 2D_2 \subseteq \Bbb Z_4^n$. In \eqref{eq1},  we define a quaternary code 
\begin{equation*}
\mathcal C_{D}=\{c_{D}(\mathbf{u})=(\mathbf{u}\cdot \mathbf{t})_{\mathbf{t}\in D}:\mathbf u\in \Bbb Z_4^n\},
\end{equation*}
where $\cdot$ is Euclidean inner product over $\Bbb Z_4$. 

If $\mathbf{u}=\mathbf{0}$, then the Lee weight of the codeword $c_{D}(\mathbf{u})$ is equal zero. Next suppose that $ \mathbf{0}\neq\mathbf{u}=\boldsymbol{\alpha}+2\boldsymbol{\beta}$, $\mathbf{t}={\bf t}_1+2{\bf t}_2$, where ${\boldsymbol{\alpha}=(\alpha_1, \ldots, \alpha_n),\boldsymbol{\beta}=(\beta_1, \ldots, \beta_n)}\in \Bbb Z_2^n$, ${\mathbf t_i}\in D_{i}$, $i = 1, 2$.  Then the Lee weight of the codeword $c_{D}(\mathbf{u})$ can be computed  by
 
\begin{eqnarray}\label{eq2}
  w_L(c_{D}(\mathbf{u}))
  &=&w_L((\boldsymbol{\alpha}+2\boldsymbol{\beta}) \cdot ({\bf t}_1+2{\bf t}_2))\nonumber\\
  &=&w_L((\boldsymbol{\alpha}{\bf t}_1+2(\boldsymbol{\alpha}{\bf t}_2+\boldsymbol{\beta}{\bf t}_1))_{{\mathbf t_i}\in D_{i}})\nonumber\\
  &=&w_H(\boldsymbol{\beta}{\bf t}_1+\boldsymbol{\alpha}{\bf t}_2) + w_H((\boldsymbol{\alpha}+\boldsymbol{\beta}){\bf t}_1+\boldsymbol{\alpha}{\bf t}_2).
\end{eqnarray}
By the definition of Hamming weight of codeword, we have 
\begin{eqnarray} \label{eq3}
w_L(c_{D}(\mathbf{u}))
&=&|D|-\sum_{\mathbf {t}_1\in  D_{1}}\sum_{\mathbf t_2\in D_{2}}(\frac12 \sum_{y\in\Bbb F_2} (-1)^{(\boldsymbol{\alpha}\mathbf{t}_2+\boldsymbol{\beta}\mathbf{t}_1)y}) \nonumber\\
&+&|D|-\sum_{\mathbf {t}_1\in  D_{1}}\sum_{\mathbf t_2\in D_{2}}(\frac12 \sum_{z\in\Bbb F_2} (-1)^{((\boldsymbol{\alpha}+\boldsymbol{\beta}){\bf t}_1 + \boldsymbol{\alpha}{\bf t}_2)z})\nonumber\\
&=&|D|-\frac12|D|-\frac12\sum_{\mathbf {t}_1\in  D_{1}}(-1)^{\boldsymbol{\beta}\mathbf{t}_1}\sum_{\mathbf t_2\in D_{2}}(-1)^{\boldsymbol{\alpha}\mathbf{t}_2} \nonumber\\
&+&|D|-\frac12|D|-\frac12\sum_{\mathbf {t}_1\in  D_{1}}(-1)^{(\boldsymbol{\alpha}+\boldsymbol{\beta}){\bf t}_1}\sum_{\mathbf t_2\in D_{2}}(-1)^{\boldsymbol{\alpha}{\bf t}_2}\nonumber\\
&=&|D|-\frac12\sum_{\mathbf t_2\in D_{2}}(-1)^{\boldsymbol{\alpha}\mathbf{t}_2}(\sum_{\mathbf {t}_1\in  D_{1}}(-1)^{\boldsymbol{\beta}\mathbf{t}_1}+\sum_{\mathbf {t}_1\in  D_{1}}(-1)^{(\boldsymbol{\alpha}+\boldsymbol{\beta}){\bf t}_1}).
\end{eqnarray} 

For $X$ a subset of $[n]$, we use $\psi(\mathbf {v}|X)$ to denote a Boolean function in $n$-variable, and $\psi(\mathbf {v}|X)=1$ if and only if $\mathbf {v}\bigcap X=\emptyset$ , which means $\mathrm{supp}({\bf v})\bigcap \mathrm{supp}({X})=\emptyset$. For a vector $\mathbf {v}=(v_1, \ldots, v_n)\in \Bbb Z_2^n$ and a  simplicial complex $\Delta_A$, by Lemma \ref{lem2.1} 
\begin{eqnarray} \label{eq4}
&&\sum_{\mathbf x\in \Delta_A} (-1)^{\mathbf {v}\cdot\bf x}=\mathcal{H}_{\Delta_A}((-1)^{v_1},(-1)^{v_2},\ldots, (-1)^{v_n})\nonumber\\
&=&\prod_{i\in A}(1+(-1)^{v_i})
=\prod_{i\in A}(2-2v_i)=2^{|A|}\prod_{i\in A}(1-v_i)=2^{|A|}\psi(\mathbf {v}|A).
\end{eqnarray} 

Before giving our main result, the following lemma is useful in determining the number  of codewords having the  same Lee weight.

\begin{lem}\label{lem31}\cite{HKWY} 
{\rm
For two subsets $A,B$ of $[n]$, we set the following cases,
$$\mathcal{U}_1 = \{{\bf u} \in \mathbb{Z}_2^{n}: {\bf u} \cap (A \cup B) = \emptyset\},$$
$$\mathcal{U}_2 = \{{\bf u} \in \mathbb{Z}_2^{n}: {\bf u} \cap A = \emptyset, {\bf u} \cap (B\backslash A) \neq \emptyset\},$$
$$\mathcal{U}_3 = \{{\bf u} \in \mathbb{Z}_2^{n}: {\bf u} \cap B = \emptyset,{\bf u} \cap(A\backslash B) \neq \emptyset\},$$
$$\mathcal{U}_4 = \{{\bf u} \in \mathbb{Z}_2^{n}: {\bf u} \cap (A\backslash B) \neq \emptyset,{\bf u} \cap(A\cap B) = \emptyset,{\bf u} \cap(B\backslash A) \neq \emptyset\},$$
$$\mathcal{U}_5 = \{{\bf u} \in \mathbb{Z}_2^{n}: {\bf u} \cap (A\backslash B) = \emptyset,{\bf u} \cap(A\cap B) \neq \emptyset,{\bf u} \cap(B\backslash A) = \emptyset\},$$
$$\mathcal{U}_6 = \{{\bf u} \in \mathbb{Z}_2^{n}: {\bf u} \cap (A\backslash B) \neq \emptyset,{\bf u} \cap(A\cap B) \neq \emptyset,{\bf u} \cap(B\backslash A) = \emptyset\},$$
$$\mathcal{U}_7 = \{{\bf u} \in \mathbb{Z}_2^{n}: {\bf u} \cap (A\backslash B) = \emptyset,{\bf u}\cap(A\cap B) \neq \emptyset,{\bf u} \cap(B\backslash A) \neq \emptyset\},$$
$$\mathcal{U}_8 = \{{\bf u} \in \mathbb{Z}_2^{n}: {\bf u} \cap (A\backslash B) \neq \emptyset,{\bf u}\cap(A\cap B) \neq \emptyset,{\bf u} \cap(B\backslash A) \neq \emptyset\}.$$
Then
$$|\mathcal{U}_1| = 2^{n-|A \cup B|},$$
$$|\mathcal{U}_2| = 2^{n-|A \cup B|}(2^{|B\backslash A|}-1),$$
$$|\mathcal{U}_3| = 2^{n-|A \cup B|}(2^{|A\backslash B|}-1),$$
$$|\mathcal{U}_4| = 2^{n-|A \cup B|}(2^{|A\backslash B|}-1)(2^{|B\backslash A|}-1),$$
$$|\mathcal{U}_5| = 2^{n-|A \cup B|}(2^{|A\cap B|}-1),$$
$$|\mathcal{U}_6| = 2^{n-|A \cup B|}(2^{|A\backslash B|}-1)(2^{|A\cap B|}-1),$$
$$|\mathcal{U}_7| = 2^{n-|A \cup B|}(2^{|A\cap B|}-1)(2^{|B\backslash A|}-1),$$
$$|\mathcal{U}_8| = 2^{n-|A \cup B|}(2^{|A\backslash B|}-1)(2^{|A\cap B|}-1)(2^{|B\backslash A|}-1).$$
}
\end{lem}

\begin{thm}\label{thm1} {\rm 
Let $n \geq 2$ be an integer and $\emptyset \neq A,B \subseteq [n]$.
Assume that $\Delta$ is a simplicial complex of $\mathbb{F}_2^n$ with the two maximal elements $A,B$ and  $\Delta_A$ is a simplicial complex of $\mathbb{F}_2^n$. 
In Equation \eqref{eq1}, let $D = \Delta_A + 2\Delta^*$, where $\Delta^*=\Delta \backslash \{{\bf 0}\}$. Then the code  $\mathcal C_{D}$ in \eqref{eq1} has length $|D|=2^{|A|}(2^{|A|}+2^{|B|}-2^{|A \cap B|}-1)$ and size $4^{|A|}2^{|B\backslash  A| }$. Its Lee weight distribution is given in Table \ref{tabZ4WD1}.

}
\end{thm}

\begin{table}[h]  
  \caption{Lee weight distribution of $\mathcal C_{D}$ in Theorem \ref{thm1} }   
  \begin{tabu} to 0.9\textwidth{X[2.5,c]|X[2,c]}  
  \hline 
  \rm{Lee Weight}&\rm{Frequency}\\ 
  \hline
  $0$&$ 1$\\ 
  \hline
  $2^{|A|}(2^{|A|}+2^{|B|}-2^{|A \cap B|}-1)$&$(2^{|A\cup B|}-2^{|B \backslash A|})(2^{|A|}-1)$\\ 
  \hline
  $2^{|A|+|B|}$&$ (2^{|B\backslash A|}-1)$\\ 
  \hline
  $2^{2|A|} + 2^{|A|+|B|-1}  - 2^{|A|-1}(2^{|A\cap B|}+1)$&$
  2(2^{|A\backslash B|}-1)
  $\\ 
  \hline
  $2^{|A|}(2^{|A|}+2^{|B|})- 2^{|A|-1}(2^{|A\cap B|}+1)$&$2(2^{|A\backslash B|}-1)(2^{|B\backslash A|}-1)$\\ 
  \hline
  $2^{|A|}(2^{|A|}+2^{|B|}-2^{|A \cap B|}) - 2^{|A|-1}$&$2^{1+|A \cup B|-|A\cap B|}(2^{|A\cap B|}-1)$\\ 
  \hline
   
  \end{tabu}
  \label{tabZ4WD1}  
\end{table}

\begin{proof}

By Equations \eqref{eq3} and \eqref{eq4}
\begin{eqnarray*} 
&&w_L(c_{D}(\mathbf{u}))\\
&=&|D|-\frac12\sum_{\mathbf t_2\in D_{2}}(-1)^{\boldsymbol{\alpha}\mathbf{t}_2}(\sum_{\mathbf {t}_1\in  D_{1}}(-1)^{\boldsymbol{\beta}\mathbf{t}_1}+\sum_{\mathbf {t}_1\in  D_{1}}(-1)^{(\boldsymbol{\alpha}+\boldsymbol{\beta}){\bf t}_1})\nonumber\\
&=&|D|-2^{|A|-1}(2^{|A|}\psi(\boldsymbol{\alpha}|A) +  2^{|B|}\psi(\boldsymbol{\alpha}|B)-2^{|A \cap B|}\psi(\boldsymbol{\alpha}|A\cap B)-1)\psi (\boldsymbol{\beta}|A)\\
&-&2^{|A|-1}(2^{|A|}\psi(\boldsymbol{\alpha}|A) +  2^{|B|}\psi(\boldsymbol{\alpha}|B)-2^{|A \cap B|}\psi(\boldsymbol{\alpha}|A\cap B)-1)\psi ((\boldsymbol{\alpha}+\boldsymbol{\beta})|A).
\end{eqnarray*} 



Next we should consider the following cases:

Case $1$, $\boldsymbol{\alpha} = \boldsymbol{\beta} = {\bf 0}$. We have $w_L(c_{D}(\mathbf{u}))=0$.

Case $2$, $\boldsymbol{\alpha}  = {\bf 0}, \boldsymbol{\beta}\neq {\bf 0}$. Then 
\begin{eqnarray*}
w_L(c_{D}(\mathbf{u}))&=&|D|-2^{|A|}( 2^{|A|} + 2^{|B|} -  2^{|A \cap B|}-1)\psi (\boldsymbol{\beta}|A) \\
&=&
  \begin{cases}
    0&if\ \boldsymbol{\beta} \cap A = \emptyset,\\
    2^{|A|}(2^{|A|}+2^{|B|}-2^{|A \cap B|}-1)&if\ \boldsymbol{\beta} \cap A \neq \emptyset.
  \end{cases}
\end{eqnarray*}
In this case, when $\boldsymbol{\beta} \cap A = \emptyset$ we have the number of the $\boldsymbol{\beta}$ is $2^{n-|A|}-1$; when $\boldsymbol{\beta} \cap A \neq \emptyset$ the number of such $\boldsymbol{\beta}$ is $2^n - 2^{n-|A|}$. 

Case $3$, $\boldsymbol{\alpha}\neq {\bf 0}, \boldsymbol{\beta}= {\bf 0}$. We have
\begin{eqnarray*}
&&w_L(c_{D}(\mathbf{u}))\\
&=&|D|-2^{|A|-1}( 2^{|A|}\psi(\boldsymbol{\alpha}|A) +  2^{|B|}\psi(\boldsymbol{\alpha}|B) -  2^{|A \cap B|}\psi(\boldsymbol{\alpha}|A\cap B)-1)(1 + \psi (\boldsymbol{\alpha}|A)) \\
&=&
  \begin{cases}
    0&if\ \boldsymbol{\alpha} \in \mathcal{U}_1,\\
    2^{|A|+|B|}&if\ \boldsymbol{\alpha} \in \mathcal{U}_2,\\
    2^{|A| + |A|} + 2^{|A|+|B|-1} - 2^{|A|+|A \cap B|-1} - 2^{|A|-1}&if\ \boldsymbol{\alpha} \in \mathcal{U}_3,\\
    2^{|A|}(2^{|A|}+2^{|B|})-2^{|A|+|A \cap B|-1}- 2^{|A|-1}&if\ \boldsymbol{\alpha} \in \mathcal{U}_4,\\
    2^{|A|}(2^{|A|}+2^{|B|}-2^{|A \cap B|}) - 2^{|A|-1}&if\ \boldsymbol{\alpha} \in \mathcal{U}_i, i\in\{5,6,7,8\}.
  \end{cases}
\end{eqnarray*}

By Lemma \ref{lem31}, we have the numbers of such $\boldsymbol{\alpha}$ in $\mathcal{U}_i$, $i\in\{1,2, \ldots ,8\}$.


Case $4$, $\boldsymbol{\alpha} = \boldsymbol{\beta}\neq {\bf 0} $. We have
\begin{eqnarray*}&&w_L(c_{D}(\mathbf{u}))\\
&=&|D|-2^{|A|-1}( 2^{|A|}\psi(\boldsymbol{\alpha}|A) +  2^{|B|}\psi(\boldsymbol{\alpha}|B) -  2^{|A \cap B|}\psi(\boldsymbol{\alpha}|A\cap B)-1)(1 + \psi (\boldsymbol{\beta}|A))\\
&=&
  \begin{cases}
    0&if\ \boldsymbol{\alpha} \in \mathcal{U}_1,\\
    2^{|A|+|B|}&if\ \boldsymbol{\alpha} \in \mathcal{U}_2,\\
    2^{|A| + |A|} + 2^{|A|+|B|-1} - 2^{|A|+|A \cap B|-1} - 2^{|A|-1}&if\ \boldsymbol{\alpha} \in \mathcal{U}_3,\\
    2^{|A|}(2^{|A|}+2^{|B|})-2^{|A|+|A \cap B|-1}- 2^{|A|-1}&if\ \boldsymbol{\alpha} \in \mathcal{U}_4,\\
    2^{|A|}(2^{|A|}+2^{|B|}-2^{|A \cap B|}) - 2^{|A|-1}&if\ \boldsymbol{\alpha} \in \mathcal{U}_i, i\in\{5,6,7,8\}.
  \end{cases}
\end{eqnarray*}
Similarly, using Lemma \ref{lem31} we could find the numbers of $\boldsymbol{\alpha}$.

Case $5$, $\boldsymbol{\alpha} \neq \boldsymbol{\beta}\neq {\bf 0} $. We have the following four subcases:
\begin{enumerate}
\item  When $\boldsymbol{\beta} \cap A \neq \emptyset$, $(\boldsymbol{\alpha}+\boldsymbol{\beta}) \cap A \neq \emptyset$. We have
$w_L(c_{D}(\mathbf{u})) = 2^{|A|}(2^{|A|}+2^{|B|}-2^{|A \cap B|}-1)$. 
Note that $\boldsymbol{\beta} \cap A \neq \emptyset$ and $(\boldsymbol{\alpha}+\boldsymbol{\beta}) \cap A \neq \emptyset$ which is equivalent  to 
$\boldsymbol{\beta} \cap A \neq \emptyset $ and $ \boldsymbol{\alpha} \cap A \neq \emptyset$, the number of such $(\boldsymbol{\alpha},\boldsymbol{\beta})$ is $$(2^n-2^{n-|A|})(2^n-2^{n-|A|}-1).$$

\item When $\boldsymbol{\beta} \cap A = \emptyset$, $(\boldsymbol{\alpha}+\boldsymbol{\beta}) \cap A \neq \emptyset$. Then $w_L(c_{D}(\mathbf{u}))$ could have:
\begin{equation*}
  \begin{cases}
    2^{|A| + |A|} + 2^{|A|+|B|-1} - 2^{|A|+|A \cap B|-1} - 2^{|A|-1}&if\ \boldsymbol{\alpha} \in \mathcal{U}_3,\\
    2^{|A|}(2^{|A|}+2^{|B|})-2^{|A|+|A \cap B|-1}- 2^{|A|-1}&if\ \boldsymbol{\alpha} \in \mathcal{U}_4,\\
    2^{|A|}(2^{|A|}+2^{|B|}-2^{|A \cap B|}) - 2^{|A|-1}&if\ \boldsymbol{\alpha} \in \mathcal{U}_i, i\in\{5,6,7,8\}.
  \end{cases}
\end{equation*}
Note that $\boldsymbol{\beta} \cap A = \emptyset$ and $(\boldsymbol{\alpha}+\boldsymbol{\beta}) \cap A \neq \emptyset$ which is equivalent  to $\boldsymbol{\beta} \cap A = \emptyset $ and $ \boldsymbol{\alpha} \cap A \neq \emptyset$, so such $\boldsymbol{\alpha}$ can only in $\mathcal{U}_i$, $i \in \{3,4, \ldots ,8 \}$. Then the number of such $(\boldsymbol{\alpha},\boldsymbol{\beta})$ is $$2^{n-|A \cup B|}(2^{|A\backslash B|}-1)(2^{n-|A|}-1),$$
when $\boldsymbol{\alpha} \in \mathcal{U}_3$;   
the number of such $(\boldsymbol{\alpha},\boldsymbol{\beta})$ is $$2^{n-|A \cup B|}(2^{|A\backslash B|}-1)(2^{|B\backslash A|}-1)(2^{n-|A|}-1),$$when $\boldsymbol{\alpha} \in \mathcal{U}_4$; and the number of such $(\boldsymbol{\alpha},\boldsymbol{\beta})$ is $$(2^n-2^{n-|A \cap B|})(2^{n-|A|}-1),$$
when $\boldsymbol{\alpha} \in \mathcal{U}_i$, $i\in\{5,6,7,8\}$.

\item When $\boldsymbol{\beta} \cap A \neq \emptyset$, $(\boldsymbol{\alpha}+\boldsymbol{\beta}) \cap A = \emptyset$. Then $w_L(c_{D}(\mathbf{u}))$ could have:
\begin{equation*}
  \begin{cases}
    2^{|A| + |A|} + 2^{|A|+|B|-1} - 2^{|A|+|A \cap B|-1} - 2^{|A|-1}&if\ \boldsymbol{\alpha} \in \mathcal{U}_3,\\
    2^{|A|}(2^{|A|}+2^{|B|})-2^{|A|+|A \cap B|-1}- 2^{|A|-1}&if\ \boldsymbol{\alpha} \in \mathcal{U}_4,\\
    2^{|A|}(2^{|A|}+2^{|B|}-2^{|A \cap B|}) - 2^{|A|-1}&if\ \boldsymbol{\alpha} \in \mathcal{U}_i, i\in\{5,6,7,8\}.
  \end{cases}
\end{equation*}
Note that $\boldsymbol{\beta} \cap A \neq \emptyset$ and $(\boldsymbol{\alpha}+\boldsymbol{\beta}) \cap A = \emptyset$ which is equivalent  to $\boldsymbol{\beta} \cap A \neq \emptyset $ and $ \boldsymbol{\alpha} \cap A \neq \emptyset$ and $\boldsymbol{\beta} \cap A = \boldsymbol{\alpha} \cap A$, so such $\boldsymbol{\alpha}$ can only in $\mathcal{U}_i$, $i \in \{3,4, \ldots ,8 \}$. Then the number of such $(\boldsymbol{\alpha},\boldsymbol{\beta})$ is $$2^{n-|A \cup B|}(2^{|A\backslash B|}-1)(2^{n-|A|}-1),$$
when $\boldsymbol{\alpha} \in \mathcal{U}_3$; the number of such $(\boldsymbol{\alpha},\boldsymbol{\beta})$ is $$2^{n-|A \cup B|}(2^{|A\backslash B|}-1)(2^{|B\backslash A|}-1)(2^{n-|A|}-1),$$ 
when $\boldsymbol{\alpha} \in \mathcal{U}_4$; and the number of such $(\boldsymbol{\alpha},\boldsymbol{\beta})$ is 
$$2^{n- |B|}(2^{|A \cap B|}-1)(2^{n-|A \cup B|}(2^{|A\backslash B|}-1)(2^{|B\backslash A|}-1)+(2^{n-|A|}-1)),$$
when $\boldsymbol{\alpha} \in \mathcal{U}_i$, $i\in\{5,6,7,8\}$.

\item When $\boldsymbol{\beta} \cap A = \emptyset$, $(\boldsymbol{\alpha}+\boldsymbol{\beta}) \cap A = \emptyset$. Then $w_L(c_{D}(\mathbf{u}))$ could have:
\begin{equation*}
  \begin{cases}
    0&if\ \boldsymbol{\alpha} \in \mathcal{U}_1,\\
    2^{|A|+|B|}&if\ \boldsymbol{\alpha} \in \mathcal{U}_2.
  \end{cases}
\end{equation*}
Note that $\boldsymbol{\beta} \cap A = \emptyset$ and $(\boldsymbol{\alpha}+\boldsymbol{\beta}) \cap A = \emptyset$ which is equivalent  to $\boldsymbol{\beta} \cap A = \emptyset $ and $ \boldsymbol{\alpha} \cap A \neq \emptyset$, so such $\boldsymbol{\alpha}$ can only in $\mathcal{U}_i$, $i \in \{1,2\}$. Then the number of such $(\boldsymbol{\alpha},\boldsymbol{\beta})$ is $$(2^{n-|A \cup B|}-1)(2^{n-|A|}-2),$$ 
when $\boldsymbol{\alpha} \in \mathcal{U}_1$; the number of such $(\boldsymbol{\alpha},\boldsymbol{\beta})$ is $$2^{n-|A \cup B|}(2^{|B\backslash A|}-1)(2^{n-|A|}-2),$$
when $\boldsymbol{\alpha} \in \mathcal{U}_2$. 

\end{enumerate}

Then we get the Lee weight distribution of the code. Next we will determine  the size of the quaternary code. 

Let $D=\{ {\bf g}_1,{\bf g}_2, \ldots, {\bf g}_{|D|} \}\subseteq \mathbb Z_4^n$. Then  $\mathcal{C}_{D}$ of length $|D|$  can be defined by\begin{equation*}\mathcal{C}_{D}= \{c_{\bf u}=({\bf u\cdot g}_1, {\bf u\cdot  g}_2, \ldots, {\bf u\cdot { g}}_{|D|}): {\bf  u}\in  \mathbb Z_4^n\}, \end{equation*}  where $\cdot $ denotes the Euclidean inner product of two elements in $\mathbb Z_4^n.$  Let $G$ be the $n\times |D|$ matrix as follows:
\begin{equation*}G=[{\bf g}_1^T \; {\bf g}_2^T \; \cdots \; {\bf g}_{|D|}^T],\end{equation*}
where the column vector ${\bf g}_i^T$ denotes the transpose of a row vector ${\bf g}_i$. Note that $D = \Delta_A + 2\Delta^*$ and $\Delta_A \subseteq \Delta$. From (2.1), it is easy to check that $k_2=|B\backslash A|$ and hence $k_1=|A|$ due to the fact that the number of codewords in $\mathcal C_D$ is $2^{|A\cup B|+|A|}$. Then the code $\mathcal{C}_{D}$ is of size $4^{|A|}2^{|B\backslash A|}$.

This completes the proof.
\end{proof}

\begin{cor}\label{cor3.3}{\rm  
Notations are the same as in Theorem \ref{thm1}.
If $A \cap B = \emptyset$. Then the code  $\mathcal C_{D}$ in \eqref{eq1} has length $|D|=2^{|A|}(2^{|A|}+2^{|B|}-2)$ and size $4^{|A|}2^{|B\backslash  A| }$. Its Lee weight distribution is given in Table \ref{tabZ4WD2}.
} 
\end{cor}

\begin{table}[h]  
  \caption{Lee weight distribution of $\mathcal C_{D}$ in Corollary 3.3 }   
  \begin{tabu} to 0.8\textwidth{X[1.7,c]|X[2,c]}  
  \hline 
  \rm{Lee Weight}&\rm{Frequency}\\ 
  \hline
  $0$&$ 1$\\ 
  \hline
  $2^{|A|}(2^{|A|}+2^{|B|}-2)$&$2^{| B|}(2^{|A|}-1)^2$\\ 
  \hline
  $2^{|A|+|B|}$&$ (2^{|B|}-1)$\\ 
  \hline
  $2^{2|A| } + 2^{|A|+|B|-1}-2^{|A|} $&$
  2(2^{|A|}-1)
  $\\ 
  \hline
  $2^{|A|}(2^{|A|}+2^{|B|})-2^{|A|}$&$2(2^{|A|}-1)(2^{|B|}-1)$\\ 
  \hline

  \end{tabu}
  \label{tabZ4WD2}  
\end{table}

Here are some examples to illustrate Theorem \ref{thm1} as {follows}:

\begin{exa}\label{exa3.2}{\rm
 Let $n = 3$, $A = \{1\}$ and $B = \{2, 3\}$, $\Delta$ be generated by $A$ and $B$, and $D=\Delta_A+2\Delta^*$. Then the $\mathcal{C}_{D}$ in (1.1) is a three-Lee-weight quaternary code with length $ 8$, size $4^12^2$ and Lee minimum distance $ 6$. Its weight distribution is given by  $1+2z^{6}+7z^{8}+6z^{10}$, which is confirmed  by Magma. 
According to the online table \cite{AAD,Z4}, the best known quaternary linear code with length $ 8$ size $4^12^2$ has Lee minimum distance $8$.  
}
    
\end{exa}

\begin{exa}\label{exa3.1}{\rm
 Let $n = 3$, $A = \{1,2\}$ and $B = \{2, 3\}$, $\Delta$ be generated by $A$ and $B$, and $D=\Delta_A+2\Delta^*$. Then the $\mathcal{C}_{D}$ in (1.1) is a five-Lee-weight quaternary code with length $20$, size $4^22^1$ and Lee minimum distance $16$. Its Lee weight distribution is given by  $1+z^{16}+2z^{18}+18z^{20}+8z^{22}+2z^{26}$, which is confirmed  by Magma. According to the online table \cite{AAD,Z4}, the code is new. %

}
    
\end{exa}

\begin{exa}\label{exa3.4}{\rm
Let $n = 4$, $A = \{1, 2\}$, $B = \{3, 4\}$, $\Delta$ be generated by $A$ and $B$, and $D=\Delta_A+2\Delta^*$. Then the  $\mathcal{C}_{D}$ in (1.1) is a four-Lee-weight quaternary linear code with parameters length $24$ and size $4^22^2$, and Lee minimum distance $16$. 
Its weight distribution is given by $1+3z^{16}+6z^{20}+36z^{24}+18z^{28}$, which is confirmed  by Magma. According to the online table \cite{AAD,Z4}, the code is new. 


}
\end{exa}

\begin{exa}\label{exa3.3}{\rm
Let $n = 5$, $A = \{2, 3\}$, $B = \{3, 4, 5\}$, $\Delta$ be generated by $A$ and $B$, and $D = \Delta_A + 2\Delta^*$. Then the $\mathcal{C}_{D}$ in (1.1) is a five-Lee-weight quaternary linear code with parameters length $36$ and size $4^22^2$, and Lee minimum distance $26$. 
Its weight distribution is given by $1+2z^{26}+3z^{32}+36z^{36}+16z^{38}+6z^{42}$, which is confirmed  by Magma. According to the online table \cite{AAD,Z4}, the code is new.


}
\end{exa}

\begin{thm}\label{thm2} {\rm
Let $\emptyset \neq A,B \subseteq [n]$, $n \geq 2$ be an integer and $\Delta$ a simplicial complex of $\mathbb{F}_2^n$ with two maximal elements $A,B$, let $D = \Bbb Z_2^n + 2\Delta^*$. Then the $\mathcal C_{D}$ is a quaternary code of length $|D|=2^{n}(2^{|A|}+2^{|B|}-2^{|A \cap B|}-1)$ and size $4^{n}$. Its Lee weight distribution is given in Table \ref{tabthm2}.
 }
\end{thm}

\begin{table}[h]  
  \caption{Lee weight distribution of $\mathcal C_{D}$ in Theorem \ref{thm2}}   
  \begin{tabu} to 1\textwidth{X[1.7,c]|X[1.6,c]}  
  \hline 
  \rm{Lee Weight}&\rm{Frequency}\\ 
  \hline
  $0$&$1$\\ 
  \hline
  $2^{n}(2^{|A|}+2^{|B|}-2^{|A \cap B|}-1)$&$(2^{n}-1)(2^{n}-1)$\\ 
  \hline
  $2^{n-1}(2^{|A|}+2^{|B|}-2^{|A \cap B|}-1)$&$2(2^{n-|A \cup B|}-1)$\\ 
  \hline 
  $2^{n+|B|}  +2^{n-1}(2^{|A|}-2^{|A \cap B|}-1)$&$2^{1+n-|A \cup B|}(2^{|B\backslash A|}-1)$\\ 
  \hline
  $2^{n+|A|}  +2^{n-1}(2^{|B|}-2^{|A \cap B|}-1)$&$2^{1+n-|A \cup B|}(2^{|A\backslash B|}-1)$\\ 
  \hline
  $2^{n}(2^{|A|}+2^{|B|}) -2^{n-1}(2^{|A \cap B|}+1)$&$2^{1+n-|A \cup B|}(2^{|A\backslash B|}-1)(2^{|B\backslash A|}-1)$\\ 
  \hline
  $2^{n}(2^{|A|}+2^{|B|}-2^{|A \cap B|})-2^{n-1}$&$2(2^n-2^{n-|A\cap B|})$\\ 
  \hline
  
  \end{tabu}
  \label{tabthm2}  
\end{table}

\begin{proof}

By Equations \eqref{eq3} and \eqref{eq4}
\begin{eqnarray*} 
&&w_L(c_{D}(\mathbf{u}))\\
&=&|D|-\frac12\sum_{\mathbf t_2\in D_{2}}(-1)^{\boldsymbol{\alpha}\mathbf{t}_2}(\sum_{\mathbf {t}_1\in  D_{1}}(-1)^{\boldsymbol{\beta}\mathbf{t}_1}+\sum_{\mathbf {t}_1\in  D_{1}}(-1)^{(\boldsymbol{\alpha}+\boldsymbol{\beta}){\bf t}_1})\nonumber\\
&=&|D|-2^{n-1}( 2^{|A|}\psi(\boldsymbol{\alpha}|A) +  2^{|B|}\psi(\boldsymbol{\alpha}|B) -  2^{|A \cap B|}\psi(\boldsymbol{\alpha}|A\cap B)-1)(\delta_{\mathbf{0},\boldsymbol{\beta}}+\delta_{\mathbf{0},(\boldsymbol{\alpha}+\boldsymbol{\beta})} ) .
\end{eqnarray*} 

Next we should consider the following cases:

Case $1$, $\boldsymbol{\alpha} = \boldsymbol{\beta} = {\bf 0}$. We have 
$w_L(c_{D}(\mathbf{u}))=0$.

Case $2$, $\boldsymbol{\alpha}  = {\bf 0}, \boldsymbol{\beta}\neq {\bf 0}$. We have
$w_L(c_{D}(\mathbf{u}))=2^{n}(2^{|A|}+2^{|B|}-2^{|A \cap B|}-1)$.

Case $3$, $\boldsymbol{\alpha}\neq {\bf 0}, \boldsymbol{\beta}= {\bf 0}$. We have $w_L(c_{D}(\mathbf{u}))\\
=|D|-2^{n-1}( 2^{|A|}\psi(\boldsymbol{\alpha}|A) +  2^{|B|}\psi(\boldsymbol{\alpha}|B) -  2^{|A \cap B|}\psi(\boldsymbol{\alpha}|A\cap B)-1) $  
\begin{equation*}
=
  \begin{cases}
    2^{n-1}(2^{|A|}+2^{|B|}-2^{|A \cap B|}-1)&if\ \boldsymbol{\alpha} \in \mathcal{U}_1,\\
    2^{n}(2^{|A|}+2^{|B|}-2^{|A \cap B|}-1)-2^{n-1}(2^{|A|}-2^{|A \cap B|}-1)&if\ \boldsymbol{\alpha} \in \mathcal{U}_2,\\
    2^{n}(2^{|A|}+2^{|B|}-2^{|A \cap B|}-1)-2^{n-1}(2^{|B|}-2^{|A \cap B|}-1)&if\ \boldsymbol{\alpha} \in \mathcal{U}_3,\\
    2^{n}(2^{|A|}+2^{|B|}-2^{|A \cap B|}-1)+2^{n-1}(2^{|A \cap B|}+1)&if\ \boldsymbol{\alpha} \in \mathcal{U}_4,\\
    2^{n}(2^{|A|}+2^{|B|}-2^{|A \cap B|}-1)+2^{n-1}&if\ \boldsymbol{\alpha} \in \mathcal{U}_i, i\in\{5,6,7,8\}.
  \end{cases}
\end{equation*}

By Lemma \ref{lem31}, we have the numbers of such $\boldsymbol{\alpha}$ in $\mathcal{U}_i$, $i\in\{1,2, \ldots ,8\}$.

Case $4$, $\boldsymbol{\alpha} = \boldsymbol{\beta}\neq {\bf 0} $. We have
$w_L(c_{D}(\mathbf{u}))\\
=|D|-2^{n-1}( 2^{|A|}\psi(\boldsymbol{\alpha}|A) +  2^{|B|}\psi(\boldsymbol{\alpha}|B) -  2^{|A \cap B|}\psi(\boldsymbol{\alpha}|A\cap B)-1)$  
\begin{equation*}
=
  \begin{cases}
    2^{n-1}(2^{|A|}+2^{|B|}-2^{|A \cap B|}-1)&if\ \boldsymbol{\alpha} \in \mathcal{U}_1,\\
    2^{n}(2^{|A|}+2^{|B|}-2^{|A \cap B|}-1)-2^{n-1}(2^{|A|}-2^{|A \cap B|}-1)&if\ \boldsymbol{\alpha} \in \mathcal{U}_2,\\
    2^{n}(2^{|A|}+2^{|B|}-2^{|A \cap B|}-1)-2^{n-1}(2^{|B|}-2^{|A \cap B|}-1)&if\ \boldsymbol{\alpha} \in \mathcal{U}_3,\\
    2^{n}(2^{|A|}+2^{|B|}-2^{|A \cap B|}-1)+2^{n-1}(2^{|A \cap B|}+1)&if\ \boldsymbol{\alpha} \in \mathcal{U}_4,\\
    2^{n}(2^{|A|}+2^{|B|}-2^{|A \cap B|}-1)+2^{n-1}&if\ \boldsymbol{\alpha} \in \mathcal{U}_i, i\in\{5,6,7,8\}.
  \end{cases}
\end{equation*}
Similarly, the numbers of each $\boldsymbol{\alpha}$ can also be found by Lemma \ref{lem31}.

Case $5$, $\boldsymbol{\alpha} \neq \boldsymbol{\beta} \neq {\bf 0} $. We have 
$w_L(c_{D}(\mathbf{u}))=2^{n}(2^{|A|}+2^{|B|}-2^{|A \cap B|}-1)$. Then the number of such $(\boldsymbol{\alpha},\boldsymbol{\beta})$ is $(2^{n}-1)(2^{n}-2).$

{Similarly} to Theorem 3.2, it is easy to check that the code has size $4^n$. 
This completes the proof.
\end{proof}

\begin{table}[h]  
  \caption{Lee weight distribution of $\mathcal C_{D}$ in Corollary \ref{cor3}}   
  \begin{tabu} to 1\textwidth{X[1.7,c]|X[1.6,c]}  
  \hline 
  \rm{Lee Weight}&\rm{Frequency}\\ 
  \hline
  $0$&$1$\\ 
  \hline
  $2^{n}(2^{|A|}+2^{|B|}-2^{|A \cap B|}-1)$&$(2^{n}-1)(2^{n}-1)$\\ 
  \hline 
  $2^{n+|B|}  +2^{n-1}(2^{|A|}-2^{|A \cap B|}-1)$&$2^{}(2^{|B\backslash A|}-1)$\\ 
  \hline
  $2^{n+|A|}  +2^{n-1}(2^{|B|}-2^{|A \cap B|}-1)$&$2^{}(2^{|A\backslash B|}-1)$\\ 
  \hline
  $2^{n}(2^{|A|}+2^{|B|}) -2^{n-1}(2^{|A \cap B|}+1)$&$2^{}(2^{|A\backslash B|}-1)(2^{|B\backslash A|}-1)$\\ 
  \hline
  $2^{n}(2^{|A|}+2^{|B|}-2^{|A \cap B|})-2^{n-1}$&$2(2^n-2^{n-|A\cap B|})$\\ 
  \hline
  
  \end{tabu}
  \label{tabcor3}  
\end{table}

\begin{cor}
\label{cor3} {\rm In Theorem \ref{thm2}  
if $|A \cup B| = n$, then the quaternary code $\mathcal{C}_{D}$ is a five-Lee-weight code and its Lee weight distribution is given in Table \ref{tabcor3}. }
\end{cor}

\begin{cor}
\label{cor4} {\rm 
In Theorem \ref{thm2} if $|A \cup B| = n$ and $A \cap B= \emptyset $. Then the quaternary code $\mathcal{C}_{D}$ is a four-Lee-weight code and its Lee weight distribution is given in Table \ref{tabcor4}. }
\end{cor}

\begin{table}[h]  
  \caption{Lee weight distribution of $\mathcal C_{D}$ in Corollary \ref{cor4}}   
  \begin{tabu} to 1\textwidth{X[1.7,c]|X[1.6,c]}  
  \hline 
  \rm{Lee Weight}&\rm{Frequency}\\ 
  \hline
  $0$&$1$\\ 
  \hline
  $2^{n}(2^{|A|}+2^{|B|}-2)$&$(2^{n}-1)(2^{n}-1)$\\ 
  \hline 
  $2^{n+|B|}  +2^{n-1}(2^{|A|}-2)$&$2^{}(2^{|B\backslash A|}-1)$\\ 
  \hline
  $2^{n+|A|}  +2^{n-1}(2^{|B|}-2)$&$2^{}(2^{|A\backslash B|}-1)$\\ 
  \hline
  $2^{n}(2^{|A|}+2^{|B|}-1) $&$2^{}(2^{|A\backslash B|}-1)(2^{|B\backslash A|}-1)$\\ 
  \hline
  
  \end{tabu}
  \label{tabcor4}  
\end{table}

Here are some examples to illustrate Theorem \ref{thm2} as {follows}:

\begin{exa}\label{exa3.5} {\rm
Let $n = 3$, $A = \{1,2\}$, $B = \{2, 3\}$, $\Delta$ be generated by $A$ and $B$, and $D = \Bbb Z_2^n + 2\Delta^*$. Then the $\mathcal{C}_{D}$ in (1.1) is a four-Lee-weight quaternary code with  length $ 40$, size $4^3$ and Lee minimum distance $36$.
Its weight distribution is given by $1+4z^{36}+49z^{40}+8z^{44}+2z^{52}$, which is confirmed  by Magma.  According to the online table \cite{AAD,Z4}, the code is new.

}
\end{exa}

\begin{exa}{\rm
Let $n = 4$, $A = \{1,2\}$, $B = \{2,3\}$, $\Delta$ be generated by $A$ and $B$, and $D = \Bbb Z_2^n + 2\Delta^*$. Then the $\mathcal{C}_{D}$ in (1.1) is a five-Lee-weight quaternary linear code with parameters length $80$ and size $4^4$, and Lee minimum distance $40$. 
Its weight distribution is given by $1+2z^{40}+8z^{72}+225z^{80}+16z^{88}+4z^{104}$, which is confirmed  by Magma. According to the online table \cite{AAD,Z4}, the code is new. 

}
\end{exa}





\section{ Binary images and Examples}

Recall the Gray map $\phi$ introduced in Section $2$, which is  a distance-preserving map from $(\mathbb{Z}_4^n , d_L)$ to $(\mathbb{Z}_2^{2n} , d_H)$ and is also a weight-preserving map.  Based on the Gray map $\phi$, {in} this section we will determine when the quaternary codes obtained in Section 3 have linear Gray images. Moreover, we will also present several minimal binary codes  and  provide examples of minimal binary code and their induced secret sharing schemes.

For a binary code of length $n$, size $M$ and Hamming minimum distance $d$,
we usually denote it by an $(n, M, d)$ code.

\begin{thm}\label{thm3}{\rm
Assume that $\mathcal{C}_{D}$ is as in Theorem \ref{thm1}. Then  $\phi (\mathcal{C}_{D})$ is a binary code with parameters
$$(2^{|A|+1}(2^{|A|}+2^{|B|}-2^{|A \cap B|}-1), 2^{|A \cup B| + |A|},d),$$
 where
\begin{equation*}
d=
  \begin{cases}
    2^{|A|+|B|}&{ if} \ B\not\subseteq A,  |A| \geq |B|,\\
       &  {or}\ B\subset A ,\\
    
    2^{2|A|}+2^{|A|-1}(2^{|B|}-2^{|A \cap B|}-1)&{ if} \ B\not\subseteq A,  |A| < |B|,\\

  2^{|A|}(2^{|A|}+2^{|B|}-2^{|A \cap B|}-1)  & if B  = A.\\
    
  \end{cases}
\end{equation*}

}
\end{thm}

\begin{proof} In Theorem \ref{thm1}, the parameters of the quaternary code $\mathcal{C}_{D}$ are all determined. The Gray map $\phi$ is a weight-preserving one, hence the Hamming minimum distance  of $\phi (\mathcal{C}_{D})$  is same as the Lee minimum distance of $\mathcal{C}_{D}$. Namely, for any nonzero codeword ${\bf c}={\bf c}_1+2{\bf c}_2\in \mathcal{C}_{D}$, then $w_L({\bf c})=w_H({\bf c}_2)+w_H({\bf c}_1+{\bf c}_2)=w_H(({\bf c}_2, {\bf c}_1+{\bf c}_2))=w_H(\phi({\bf c}))$.  However, the length of the binary code is doubled compared to the quaternary code. This completes the proof.
\end{proof}

Similarly to Theorem \ref{thm3}, we have the following result.

\begin{thm}\label{thm4}{\rm
Assume that $\mathcal{C}_{D}$ is as in Theorem \ref{thm2}. 

(1) If $|A\cup B| <n$, then  $\phi (\mathcal{C}_{D})$ is a binary code with parameters 
$$(2^{n+1}(2^{|A|}+2^{|B|}-2^{|A \cap B|}-1),2^{2n},2^{n-1}(2^{|A|}+2^{|B|}-2^{|A \cap B|}-1)).$$

(2) If $|A\cup B| =n$, then  $\phi (\mathcal{C}_{D})$ is a binary  code with parameters 
$$(2^{n+1}(2^{|A|}+2^{|B|}-2^{|A \cap B|}-1),2^{2n},d),$$
 where
\begin{equation*}
d=
  \begin{cases}
  2^{n}(2^{|A|}+2^{|B|}-2^{|A \cap B|}-1)& if A=[n] ~or~ B=[n],\\
  
    2^{n-1}(2^{|B|+1}+2^{|A|}-2^{|A \cap B|}-1)&{if}\  0<|B| \le  |A| <n,\\
    2^{n-1}(2^{|A|+1}+2^{|B|}-2^{|A \cap B|}-1)&{if}\ 0< |A| \le  |B| <n.\\
  \end{cases}
\end{equation*}

}
\end{thm}

 In general, determining whether the Gray image of a quaternary code is linear or not can be a challenging task. However, the following lemma plays a significant role in addressing this matter.

\begin{lem}\label{lem:linear}{\rm \cite[Corollary 3.17] {WQ}
Let $\mathcal{C}$ be a  quaternary linear code, ${\bf x}_1, \cdots, {\bf x}_m$ be a set of generators
of $\mathcal{C}$, and $C = \phi(\mathcal{C})$. Then $C$ is linear if and only if $$2\alpha({\bf x}_i) \ast \alpha({\bf x}_j) \in \mathcal{C}$$ for all
$i, j$ satisfying $1 \le i \le j \le m$, where $\alpha: \Bbb Z_4 \to \mathbb Z_2$ is a group homomorphism and $\alpha(0)=0,\alpha(1)=1, \alpha(2)=0,$ and $ \alpha(3)=1$, and $\ast$ denotes the componentwise product
of two vectors.}

\end{lem}

Now we are ready  to determine when the Gray image of a quaternary code above is linear or not. 

\begin{thm}\label{thm: linear}
{\rm Assume that $\mathcal{C}_{D}$ with $D=\Delta_A+2 \Delta^*$ is as in Theorem \ref{thm1}. Then $\phi(\mathcal{C}_{D})$ is linear if and only if $|A|=1$.
}
\end{thm}
\begin{proof} Let $D=\{ {\bf g}_1, {\bf g}_2,\ldots, {\bf g}_{|D|} \}\subseteq \mathbb Z_4^n$. Then the code $\mathcal{C}_{D}$ of length $|D|$ over $\mathbb Z_{4}$ can be defined by\begin{equation*}\mathcal{C}_{D}= \{c_{\bf u}=({\bf u\cdot g}_1, {\bf u\cdot  g}_2, \ldots, {\bf u\cdot { g}}_{|D|}): {\bf  u}\in  \mathbb Z_4^n\}, \end{equation*}  where $\cdot $ denotes the Euclidean inner product of two elements in $\mathbb Z_4^n.$  Let $G$ be the $n\times |D|$ matrix as follows:
\begin{equation*}G=[{\bf g}_1^T \; {\bf g}_2^T \; \cdots \; {\bf g}_{|D|}^T]=\left( \begin{array}{ccccccc}
{\bf x}_1 \\
\vdots\\
{\bf x}_n
\end{array} \right),\end{equation*}
where the column vector ${\bf g}_i^T$ denotes the transpose of a row vector ${\bf g}_i$. Hence $ {\bf x}_1, \ldots, {\bf x}_n $ can be viewed as a set of generators of $\mathcal{C}_{D}$.

By Lemma \ref{lem:linear},  $\phi(\mathcal{C}_{D})$ is linear if and only if $$2\alpha({\bf x}_i) \ast \alpha({\bf x}_j) \in \mathcal{C}_{D}$$ for all
$i, j$ satisfying $1 \le i \le j \le n$, where $\alpha: \Bbb Z_4 \to \mathbb Z_2$ is a group homomorphism and $\alpha(0)=0,\alpha(1)=1, \alpha(2)=0,$ and $ \alpha(3)=1$, and $\ast$ denotes the componentwise product
of two vectors.

($\Longleftarrow$). Suppose that $|A|=1$. By Theorem \ref{thm1}, the code $\mathcal{C}_{D}$ has size $4^12^{|B\backslash A|}$, which is permutation-equivalent to a code with generator matrix $G$ of the form
\begin{eqnarray} \left( \begin{array}{ccccccc}
I_{1} &A_1&  B_1+2B_2 \\
0& 2I_{k_2}&2A_2
\end{array} \right),
\end{eqnarray}
where $A_1 , A_2 , B_1$, and $B_2$ are matrices with entries $0$ or $1$. Namely, in the set $\{ {\bf x}_1, {\bf x}_2,\ldots, {\bf x}_n \}$ there is only one vector ${\bf x}_i$ such that $\alpha({\bf x}_i)\neq {\bf 0}$. Hence for any $1\le i\neq j\le n$, we always have $2\alpha({\bf x}_i) \ast \alpha({\bf x}_j)={\bf 0} \in \mathcal{C}_{D}$. It is sufficient to check that $2\alpha({\bf x}_i) \ast \alpha({\bf x}_i) \in \mathcal{C}_{D}$.  By the definition and the fact that $\mathcal{C}_{D}$ is a linear quaternary code, we have that $2\alpha({\bf x}_i) \ast \alpha({\bf x}_i)=2{\bf x}_i \in \mathcal{C}_{D}$ and hence its Gray image $\phi(\mathcal{C}_{D})$ is linear.




($\Longrightarrow$).  Assume that $\phi(\mathcal{C}_{D})$ is linear and $|A| \ge 2$. So we can select two distinct integers $l,k\in A\subseteq [n]$. Notice that ${\bf g}_i\in D =\Delta_A+2\Delta^*$ and $\Delta_A\subset \Delta^*$. Let ${\bf g}_i=({\bf g}_i^{(1)}, \ldots, {\bf g}_i^{(n)})$.
Therefore the size of the set $$\{{\bf g}_i: 1\le i\le |D|, {\bf g}_i^{(l)} \equiv {\bf g}_i^{(k)} \equiv 1\pmod 2\}$$ is equal to $(2^{|A|}+2^{|B|}-2^{|A \cap B|}-1)2^{|A|-2}$. Then we have that the vector $$2\alpha({\bf x}_l) \ast \alpha({\bf x}_k) \in \mathcal{C}_{D}$$ 
and has Lee weight $(2^{|A|}+2^{|B|}-2^{|A \cap B|}-1)2^{|A|-1}$ due to the Lee weight of 2 is 2.

In Theorem \ref{thm1},  we determined the Lee weight distribution of $\mathcal{C}_{D}$. By checking Table 3.1, the Lee weight $(2^{|A|}+2^{|B|}-2^{|A \cap B|}-1)2^{|A|-1}$  cannot be the Lee weight of a certain codeword of $\mathcal{C}_{D}$ except for one situation that $A\cap B=\emptyset$, $|A|=2$, and $|B|=1$.  Without loss of generality, assume that $n=3$, $A=\{1,2\}$, and $B=\{3\}$. The generator matrix of the code $\mathcal{C}_{D}$ is as follows
$$
\begin{bmatrix}
2 & 0 & 2 & 0  & 3 & 1 & 3 & 1  & 2 & 0 & 2 & 0 & 3 & 1 & 3 & 1  \\
0 & 2 & 2 & 0  & 0 & 2 & 2 & 0  & 1 & 3 & 3 & 2  & 1 & 3 & 3 & 1  \\
0 & 0 & 0 & 2  & 0 & 0 & 0 & 2  & 0 & 0 & 0 & 2  & 0 & 0 & 0 & 2  \\ 
\end{bmatrix}.$$ If $$2\alpha({\bf x}_1) \ast \alpha({\bf x}_2)=(0, 0,0,0, 0, 0,0,0,0, 0,0,0,2,2,2,2) \in \mathcal{C}_{D},$$ then $${\bf x}_3+2\alpha({\bf x}_1) \ast \alpha({\bf x}_2)=(0, 0,0,2, 0, 0,0,2,0, 0,0,2,2,2,2,0)\in \mathcal{C}_{D}$$ has Lee weight 12, however the Lee weight distribution  the code is $1+z^8+24z^{16}+6z^{20}$ by Corollary 3.3.  There is  a contradiction and in this case we should have that $|A|=1$.


This completes the proof.
\end{proof}

Similarly to Theorem 4.4, we have the following result.

\begin{thm}{\rm  Assume that $\mathcal{C}_{D}$ with $D=\Bbb Z_2^n+2 \Delta^*$ is as in Theorem \ref{thm2}. If $|A|>1 \mbox{ or }  |B|>1$, then $\phi(\mathcal{C}_{D})$ is not linear.

}

\end{thm} 

\begin{proof}  
By the condition,  then there exists a maximal element $E\subseteq [n]$ with $|E|>1$ such that $\Delta_E\subseteq  \Delta^*$ and $$D=\Bbb Z_2^n+2 \Delta^*=(\Delta_E+2 \Delta^*)\cup(\Bbb Z_2^n\backslash \Delta_E +2 \Delta^*)=D_1\cup D_2.$$
By Theorem \ref{thm: linear},  $\phi(\mathcal C_{D_1})$ is not linear. The code $\mathcal{C}_{D}$ can be viewed as an extended code of $\mathcal C_{D_1}$ with $|D_2|$ digits.
Hence $\phi(\mathcal{C}_{D})$ is not linear and we are done.\end{proof} 






\begin{prop}
\label{prop2}{\rm

In Theorem \ref{thm3},  if $1=|A| <|B|$, then the Gray image $\phi (\mathcal{C}_{D})$ is a binary minimal  linear code.

}
\end{prop}

\begin{proof}
 By Lemma \ref{lem2.2} and Theorem \ref{thm: linear}, it is {sufficient} to check 
\begin{eqnarray*} 
\frac{wt_{min}}{wt_{max}}
&=& \frac{2^{2|A|}+2^{|A|-1}(2^{|B|}-2^{|A \cap B|}-1)}{2^{|A|}(2^{|A|}+2^{|B|})-2^{|A|-1}(2^{|A \cap B|}+1)}\\
&=& \frac{2^{|A|+1}-2^{|A \cap B|}-1+2^{|B|}}{2^{|A|+1}-2^{|A \cap B|}-1+2^{|B|+1}}\\
&=& \frac{\lambda+2^{|B|}}{\lambda+2^{|B|+1}} > \frac{1}{2},
\end{eqnarray*} 
where $\lambda = 2^{|A|+1}-2^{|A \cap B|}-1$.

This completes the proof.
\end{proof}








Here are some examples of  binary codes and secret sharing schemes  from the Gray map, which are all confirmed  by Magma

\begin{exa}{\rm
  
  (1) In Example \ref{exa3.2}, the generator matrix of the code $\mathcal{C}_{D}$ is as follows
$$
\begin{bmatrix}
0 & 0 & 0 & 2 & 1 & 1 & 1 & 3 \\
2 & 0 & 2 & 0 & 2 & 0 & 2 & 0 \\
0 & 2 & 2 & 0 & 0 & 2 & 2 & 0 
\end{bmatrix} 
.$$ Then $\phi (\mathcal{C}_{D})$  is a binary minimal  linear code with parameters $[16, 4, 6]$.

(2) In Example \ref{exa3.1},  The generator matrix of the code $\mathcal{C}_{D}$ is as follows
$$
\begin{bmatrix}
2 & 0 & 2 & 0 & 0 & 3 & 1 & 3 & 1 & 1 & 2 & 0 & 2 & 0 & 0 & 3 & 1 & 3 & 1 & 1 \\
0 & 2 & 2 & 0 & 2 & 0 & 2 & 2 & 0 & 2 & 1 & 3 & 3 & 2 & 3 & 1 & 3 & 3 & 1 & 3 \\
0 & 0 & 0 & 2 & 2 & 0 & 0 & 0 & 2 & 2 & 0 & 0 & 0 & 2 & 2 & 0 & 0 & 0 & 2 & 2 \\ 
\end{bmatrix} =\begin{bmatrix}
{\bf x}_1  \\
{\bf x}_2  \\
{\bf x}_3   
\end{bmatrix} 
.$$ Then $w_L(2\alpha({\bf x}_1) \ast \alpha({\bf x}_2))=10$ and $2\alpha({\bf x}_1) \ast \alpha({\bf x}_2) \not\in \mathcal{C}_{D}$. Hence  $\phi (\mathcal{C}_{D})$  is a  binary nonlinear code with parameters $(40, 2^5, 16)$.
   


}   
\end{exa}


\begin{exa}{\rm

Let $n = 4$, $A = \{1\}$, $B = \{2\}$, $\Delta$ be generated by $A$ and $B$, and $D=\Delta_A+2\Delta^*$. 
   Then $\phi (\mathcal{C}_{D})$ is a minimal binary code with parameters $[8, 3, 4]$. Furthermore, see \cite{G2} it is optimal and its the dual code  $(\phi (\mathcal{C}_{D}))^{\bot}$ has parameters $[8, 5, 2]$ and generator matrix
$$
\begin{bmatrix}
1 & 0 & 0 & 0 & 0 & 0 & 1 & 1\\
0 & 1 & 0 & 0 & 0 & 0 & 1 & 1\\
0 & 0 & 1 & 0 & 0 & 1 & 1 & 0\\
0 & 0 & 0 & 1 & 0 & 1 & 1 & 0\\
0 & 0 & 0 & 0 & 1 & 1 & 1 & 1 
\end{bmatrix} 
.$$
   

In this secret sharing scheme, $7$ participants and a dealer are involved, there are altogether $12$ minimal access sets, 
$\{1,2,3\}$, $\{3,5,7\}$, $\{3,4,6\}$, $\{2,5,7\}$, $\{2,4,6\}$, $\{2,3,6,7\}$, $\{2,3,4,5\}$, $\{1,4,5,6,7\}$, $\{1,3,5,6\}$, $\{1,3,4,7\}$, $\{1,2,5,6\}$, $\{1,2,4,7\}$.  There is no dictatorial participant in this secret sharing scheme, as no participant is involved in every minimal access set. In other words, each participant is involved in the scheme.

}
\end{exa}

\begin{exa}{\rm

Let $n = 4$, $A = \{1\}$, $B = \{3,4\}$, $\Delta$ be generated by $A$ and $B$, and $D=\Delta_A+2\Delta^*$. 
   Then $\phi (\mathcal{C}_{D})$ is a minimal binary code  with parameters $[16, 4, 6]$, and its dual code $(\phi (\mathcal{C}_{D}))^{\bot}$ with parameters $[16, 12, 2]$ has generate matrix    
$$
\begin{bmatrix}
1 & 0 & 0 & 0 & 0 & 0 & 0 & 0 & 0 & 0 & 0 & 0 & 0 & 0 & 1 & 1 \\
0 & 1 & 0 & 0 & 0 & 0 & 0 & 0 & 0 & 0 & 0 & 0 & 0 & 0 & 1 & 1 \\
0 & 0 & 1 & 0 & 0 & 0 & 0 & 0 & 0 & 0 & 0 & 0 & 0 & 1 & 0 & 1 \\
0 & 0 & 0 & 1 & 0 & 0 & 0 & 0 & 0 & 0 & 0 & 0 & 0 & 1 & 0 & 1 \\
0 & 0 & 0 & 0 & 1 & 0 & 0 & 0 & 0 & 0 & 0 & 1 & 0 & 0 & 0 & 1 \\
0 & 0 & 0 & 0 & 0 & 1 & 0 & 0 & 0 & 0 & 0 & 1 & 0 & 0 & 0 & 1 \\
0 & 0 & 0 & 0 & 0 & 0 & 1 & 0 & 0 & 0 & 0 & 1 & 0 & 1 & 0 & 0 \\
0 & 0 & 0 & 0 & 0 & 0 & 0 & 1 & 0 & 0 & 0 & 1 & 0 & 1 & 0 & 0 \\
0 & 0 & 0 & 0 & 0 & 0 & 0 & 0 & 1 & 0 & 0 & 1 & 0 & 1 & 0 & 1 \\
0 & 0 & 0 & 0 & 0 & 0 & 0 & 0 & 0 & 1 & 0 & 1 & 0 & 1 & 1 & 0 \\
0 & 0 & 0 & 0 & 0 & 0 & 0 & 0 & 0 & 0 & 1 & 1 & 0 & 0 & 1 & 1 \\
0 & 0 & 0 & 0 & 0 & 0 & 0 & 0 & 0 & 0 & 0 & 0 & 1 & 1 & 1 & 1 
\end{bmatrix} 
.$$ 
}
\end{exa}

In this secret sharing scheme, there are $15$ participants and a dealer involved. All the codewords in the dual code $(\phi (\mathcal{C}_{D}))^{\bot}$ are calculated by the program. There are a total of $229$ minimal access sets, each requiring at least $10$ members to recover the secret. With the assistance of a program, we can obtain all of the access sets. Additionally, there is no dictatorial participant in this secret sharing scheme.

\section{Code Comparisons and Concluding Remarks}

To show signiﬁcant advantages of our quaternary codes, in Table 6 we {list} some quaternary codes of two maximal elements $A, B$ when $n=3$ and $n=4$. Here 
$*$ indicates that the corresponding codes are optimal, and “new” are also indicated according to the current $\Bbb Z_4$ database \cite{AAD,Z4}, and $d_L^{best}$ means the best known minimum Lee  distance of quaternary  linear codes with the same length and size of $\mathcal C_{D}$, and $d_H^{best}$ means the best known minimum Hamming  distance of binary linear codes with the same length and size {\color[rgb]{0,0,1} of   $\phi(\mathcal C_{D})$, according to codetables \cite{G2}}. It is believed that our linear quaternary codes include more new codes although we cannot compare with the parameters in $\Bbb Z_4$ database \cite{AAD, Z4} because our code length can be large. 

The main contributions  in this paper are the following

\begin{itemize}
  \item The determinations of the Lee weight distributions of the codes over $\mathbb{Z}_4$ when simplicial complexes  generated by two maximal elements in (1.1), see Theorems \ref{thm1} and \ref{thm2}.
  
  \item  Two infinite families of four-Lee-weight quaternary codes in Corollaries 3.3 and 3.10.

  \item At least nine new quaternary codes comparing the online codetable by \cite{AAD,Z4}, see Table 6.
  \item Two infinite families of binary nonlinear  codes in Theorems  \ref{thm: linear} and  4.5.
  

  \item One class of binary minimal  codes from the Gray images of the quaternary codes in Proposition \ref{prop2}.

  \item Using those  binary  minimal codes, we obtained to build secret sharing schemes in Examples 4.8 and 4.9.
\end{itemize}

It is interesting to find more new and optimal quaternary codes by using simplicial complexes and even posets.



\section*{Acknowledgments}

{\small This work was supported by the National Natural Science Foundation of China (Grant Nos. 62372247, 12101326, 61932013), the Natural Science Foundation of Jiangsu Province (Grant No. BK20210575), and the China Postdoctoral Science Foundation (Grant No. 2023M740958). The authors wish to thank the anonymous reviewers and the Associate Editor for their very helpful comments that improved the presentation and quality of this article.}

\begin{table}
  \caption{ Quaternary linear codes  from Theorems \ref{thm1} and \ref{thm2}  }   
  \begin{tabu} to 0.8\textwidth{|c|c|c|c|c|c|c|c|c|}  
  \hline 
  \rm{$n$}&\rm{$|A|$}&\rm{$|B|$}&\rm{$|A\cap B|$}&\rm{Length}&\rm{Size}&\rm{$d_L$}&\rm{ Comment }&\rm{ $\phi(\mathcal C_{D})$ }\\ 
  \hline
  3 & 1 & 1 & 0 &  4 & $4^12^1$ & 4  & new & minimal linear, $d_H^{best}=4$ \\
  \hline
  3 & 1 & 1 & 0 & 16 & $4^3$ & 8  &$d_L^{best}=16$  &  $d_H^{best}=16$  \\
  \hline
  3 & 1 & 2 & 0 & 8 & $4^12^2$ & 6 & $d_L^{best}=8$ & minimal linear, $d_H^{best}=8$ \\
  \hline
  3 & 1 & 2 & 1 & 6 & $4^12^1$ & 5 & $d_L^{best}=6$ & minimal linear, $d_H^{best}=6$ \\
  \hline
    3 & 1 & 2 & 0 & 32 & $4^3$ & 24 &$d_L^{best}=32$& nonlinear, $d_H^{best}=32$ \\
  \hline
  3 & 1 & 3 & 1 & 14 & $4^12^2$ & 9 & $d_L^{best}=14$ & minimal linear, $d_H^{best}=14$\\
  \hline
  3 & 2 & 1 & 0 & 16 & $4^22^1$ & 8 & best known * &nonlinear, $d_H^{best}=16$  \\
  \hline
  3 & 2 & 1 & 1 & 12 & $4^2$ & 6 & $d_L^{best}=12$ &nonlinear, $d_H^{best}=12$  \\
  \hline
  3 & 2 & 3 & 2 & 28 & $4^22^1$ & 22 & new & nonlinear, $d_H^{best}=28$ \\
  \hline
  3 & 2 & 2 & 1 & 20 & $4^22^1$ & 16 & new &nonlinear, $d_H^{best}=20$   \\
  \hline
  3 & 2 & 2 & 1 & 40 & $4^3$ & 36 &$d_L^{best}=40$& nonlinear, $d_H^{best}=40$ \\
  \hline
  3 & 3 & 1 & 1 & 56 & $4^3$ & 16 & $d_L^{best}=56$ & nonlinear, $d_H^{best}=56$ \\
  \hline
  3 & 3 & 2 & 2 & 56 & $4^3$ & 32 & $d_L^{best}=56$ & nonlinear, $d_H^{best}=56$ \\
  \hline

  4 & 1 & 1 & 0 & 4  & $4^12^1$ & 4 & new & minimal  linear, $d_H^{best}=4$ \\
  \hline
  4 & 1 & 3 & 0 & 16 & $4^12^3$ & 10 & $d_L^{best}=16$ & minimal linear, $d_H^{best}=16$ \\
  \hline
  4 & 1 & 3 & 1 & 14 & $4^12^2$ & 9 & $d_L^{best}=14$ & minimal linear, $d_H^{best}=14$\\
  \hline
  4 & 2 & 2 & 0 & 24 & $4^22^2$ & 16 & $d_L^{best}=24$& nonlinear, $d_H^{best}=24$ \\
  \hline
  4 & 2 & 2 & 1 & 20 & $4^22^1$ & 16 & new &nonlinear, $d_H^{best}=20$  \\
  \hline
  4 & 2 & 3 & 1 & 36 & $4^22^2$ & 26 & new &nonlinear, $d_H^{best}=35$  \\
  \hline
  4 & 2 & 3 & 2 & 28 & $4^22^1$ & 22 & new &nonlinear, $d_H^{best}=28$\\
  \hline
  4 & 2 & 4 & 2 & 60 & $4^22^2$ & 38 & new &nonlinear, $d_H^{best}=60$  \\
  \hline
  4 & 3 & 1 & 0 & 64 & $4^32^1$ & 16 & new &nonlinear, $d_H^{best}=64$  \\
  \hline
  4 & 3 & 1 & 1 & 56 & $4^3$ & 16 & $d_L^{best}=56$ &nonlinear, $d_H^{best}=56$ \\
  \hline
  4 & 3 & 2 & 1 & 72 & $4^32^1$ & 32 & $d_L^{best}=70$ &nonlinear, $d_H^{best}=71$  \\
  \hline
  4 & 3 & 2 & 2 & 56 & $4^3$ & 32 & $d_L^{best}=56$ &nonlinear, $d_H^{best}=56$  \\
  \hline
  4 & 3 & 3 & 2 & 88 & $4^32^1$ & 64 & $d_L^{best}=88$ &nonlinear, $d_H^{best}=88$  \\
  \hline
  4 & 3 & 4 & 3 & 120 & $4^3$ & 92 & $d_L^{best}=120$ &nonlinear, $d_H^{best}=120$  \\
  \hline

  \end{tabu}
  \label{tabcodesthm1}  
\end{table}

  


\end{document}